\documentclass[runningheads]{llncs}

\usepackage[T1]{fontenc}
\usepackage{graphicx}

\usepackage{booktabs}       
\usepackage{amsfonts}       
\usepackage{nicefrac}       
\usepackage{microtype}      
\usepackage{xcolor}         

\usepackage{amsthm}

\usepackage{amsmath}
\usepackage{amssymb}
\usepackage{mathtools}
\newcommand{\vect}[1]{\boldsymbol{#1}}

\usepackage{algorithm}
\usepackage{algpseudocode}

\begin{document}
\title{SBM With Multiple Samples: Improved Spectral Recovery}

%
%
\author{Sie Hendrata Dharmawan\inst{1}\orcidID{0009-0008-5078-9907} \and
Peter Chin \inst{1}\orcidID{0000-0002-1913-4223}}

%
\authorrunning{Dharmawan, Chin}
%
\institute{Dartmouth College, Hanover NH 03755, USA}
%
\maketitle              
\begin{abstract}

We study community detection in the two-block stochastic block model under the setting where multiple independent graph samples drawn from the same distribution are available. Building on a recently simplified spectral algorithm that preserves the independence of adjacency matrix entries throughout, we show that averaging $m$ independent samples before applying spectral partitioning reduces the error bound $\gamma$ exponentially in $m$: specifically, one can find a $\gamma$-correct partition with probability $1 - o(1)$ whenever $\frac{(a-b)^2}{a+b} \geq \frac{C}{m} \log \frac{2}{\gamma}$, improving the single-sample requirement by a factor of $m$. The key technical contribution is a multi-sample analogue of the spectral norm bound on the noise matrix, which propagates through the Davis–Kahan subspace angle analysis to yield the improved recovery guarantee. We provide experimental validation across a range of graph sizes ($n$ up to $1000$) and sample counts ($m$ up to $9$), demonstrating that the derived bounds are sharp and that even two or three samples yield dramatic improvements in recovery accuracy. Our results offer a rigorous theoretical foundation for graph data augmentation strategies used in modern graph representation learning.

\end{abstract}
\section{Introduction}
\label{section:intro}

Community detection in graphs has emerged as a foundational subroutine in modern graph-based learning systems, including graph neural networks (GNNs). The stochastic block model provides a canonical generative framework for benchmarking such systems: the ability of a GNN or spectral embedding layer to recover planted community structure is increasingly used as a proxy for its representational power and generalization behavior. Moreover, the multi-sample setting studied in this paper has a direct analogue in graph learning pipelines, where multiple snapshots of a dynamic network — or multiple augmented views of a graph in contrastive self-supervised learning — are aggregated before being passed to a downstream classifier. Our theoretical analysis of how averaging independent graph samples improves spectral recovery can therefore be understood as a rigorous characterization of a data augmentation strategy widely employed in graph representation learning, providing the kind of theoretical grounding that the neural network community increasingly seeks for empirically successful techniques.

The stochastic block model (SBM) serves as a prominent theoretical framework for analyzing this problem. In its simplest form, the model consists of two equal-sized blocks $V_1$ and $V_2$, each containing $n$ vertices. A random graph is generated according to the following distribution: edges between vertices within the same block occur with probability $\frac{a}{n}$, while edges between vertices in different blocks occur with probability $\frac{b}{n}$, where $a > b > 0$. Given such a graph, various algorithms exist for block recovery \cite{chin2015stochasticblockmodelcommunity}, \cite{Bui1984GraphBA}, \cite{Dyer1989TheSO}, \cite{McSherry2001SpectralPO}, \cite{CojaOghlan2009GraphPV}. 

In particular, Chin et al. introduced a spectral algorithm that achieves exponential bounds on the incorrect recovery rate in the case of a sparse graph \cite{chin2015stochasticblockmodelcommunity}. The original spectral algorithm includes a preprocessing step in which certain rows and columns of the adjacency matrix (specifically, those corresponding to vertices whose degree deviates significantly from its expectation) are deleted or modified before the eigendecomposition is performed. While this step is analytically convenient for controlling the spectral norm of the noise matrix, it introduces a subtle but consequential complication: once rows and columns are modified based on observed degree, the entries of the resulting matrix are no longer independently distributed, because whether a given entry is retained depends on the degrees of its incident vertices, which are functions of the entire row and column. This loss of independence makes it difficult to apply standard probabilistic tools (such as concentration inequalities for sums of independent random variables) to the modified matrix, and it considerably complicates any analysis that must track the joint behavior of entries across multiple samples. In \cite{dharmawan1}, Dharmawan and Chin show that this deletion-and-modification step is in fact unnecessary: the same asymptotic recovery guarantees, including the exponential bound on the incorrect recovery rate, are achievable by running the spectral algorithm directly on the unmodified adjacency matrix. By removing the preprocessing step, \cite{dharmawan1} preserves the independence of matrix entries throughout the entire algorithm, and it is precisely this independence that makes the multi-sample analysis of the present paper tractable.

In this paper, we investigate the scenario that involves multiple samples (graphs) drawn from an identical distribution. We show that each additional sample reduces the incorrect recovery rate, establish the asymptotic relationship between recovery accuracy and sample size, and provide experimental evidence for the sharpness of this relationship.

In the sparse graph case, with high probability, the linear fraction of isolated vertices \cite{Bollobás_2001} makes exact recovery impossible with high probability. In denser regimes, isolated vertices vanish and exact recovery becomes attainable, as our experiments confirm. In sparse regimes, although we cannot guarantee perfect recovery, we can still accurately recover a substantial portion of each block. Formally, we  would like to find a partition of $V_1',V_2'$ of $V = V_1 \cup V_2$ such that $V_i$ and $V_i'$ are very close to each other. To quantify the recovery accuracy, we introduce the following definition:

\begin{definition}
A collection of subsets $V_1', V_2'$ of $V_1 \cup V_2$ is $\gamma$-correct if $|Vi \cap V_i'| \geq (1-\gamma)n,i=1,2$.
\end{definition}

We would like to devise an algorithm that can guarantee $\gamma$-correctness for small $\gamma$ with high probability in polynomial time. In \cite{CojaOghlan2009GraphPV}, Coja-Oghlan proved

\begin{theorem}
\label{Coja-Oglan}
For any constant $\gamma > 0$, there exist constants $C_1, C_2 >0$ such that if $a,b > C_1$ and $\frac{(a-b)^2}{a+b} > C_2 \log(a+b)$, one can find a $\gamma$-correct partition using a polynomial time algorithm.
\end{theorem}

Coja-Oghlan proved Theorem \ref{Coja-Oglan} as part of a more general problem, and his algorithm was rather involved. In \cite{chin2015stochasticblockmodelcommunity}, Chin et. al. proved

\begin{theorem}
\label{theorem:chin-theorem}
    There are constants $C_1, C_2 >0$ such that the following holds. For any constants $a>b>C_1$ and $\gamma > 0$ satisfying
    \begin{equation}
        \frac{(a-b)^2}{a+b} \geq C_2 \log \frac{2}{\gamma}
    \end{equation}
    one can find a $\gamma$-correct partition with probability $1-o(1)$ using a simple spectral algorithm.
\end{theorem}

Theorem \ref{theorem:chin-theorem} improves the relation between the accuracy $\gamma$ and the ratio $\frac{(a-b)^2}{a+b}$. Moreover, this bound is asymptotically sharp because according to \cite{zhang2015minimaxratescommunitydetection}, there exists a constant $c > 0$ such that when 
\begin{equation}
    \label{eq:zhang}
    \frac{(a-b)^2}{a+b} \leq c \log \frac{1}{\gamma}
\end{equation}
one \textbf{cannot} recover a $\gamma$-correct partition (in expectation), regardless of the algorithm.

However, the spectral algorithm presented in \cite{chin2015stochasticblockmodelcommunity} and the corresponding proof have several steps that are rather complex. In \cite{dharmawan1}, a simpler version of the spectral algorithm is established that maintains the bound in Theorem 2. And in this paper we extend that simplified algorithm to the multi-sample setting. We show that when we draw $m$ independent samples from the same distribution, we can dramatically improve the recovery rate. 

\begin{theorem}[Main Result]
\label{theorem:main-theorem}
There exist constants $C_1, C_2 > 0$ such that the following holds. For any constants $a>b>C_1$ and $\gamma > 0$ satisfying
\begin{equation}
\label{main-equation}
    \frac{(a-b)^2}{a+b} \geq \frac{C_2}{m} \log \frac{2}{\gamma}
\end{equation}
    given $m$ independently-sampled graphs, one can find a $\gamma$-correct partition with probability $1-o(1)$.
\end{theorem}

Theorem \ref{theorem:main-theorem} allows us to dramatically improve the recovery rate by drawing additional samples from the same distribution, because Equation \ref{main-equation} implies the achievable error rate satisfies:
\begin{equation}
    \gamma \leq 2 \exp\left( -m\frac{(a-b)^2}{C_2(a+b)}\right)
\end{equation}
so each additional sample reduces $\gamma$ exponentially. 

The rest of the paper is organized as follows: in Section \ref{section:spectral_Chin} we will present the original result presented by Chin et al., including the original spectral algorithm and its main steps of proof. In Section \ref{section:spectral_us} we will present our modification and the effect of having $m$ samples has on each step of the proof. In Section \ref{section:exp_1} we will show the sharpness of some bounds in Chin's paper, while in Section \ref{section:exp_m} we will show the sharpness of the bounds in our result, especially the one asserted by Theorem \ref{theorem:main-theorem}.

\section{Original Spectral Algorithm}
\label{section:spectral_Chin}
In \cite{chin2015stochasticblockmodelcommunity}, Chin et al. gave the Spectral Algorithm that guarantees the result in Theorem \ref{theorem:chin-theorem}, which is then simplified by Dharmawan and Chin in \cite{dharmawan1}. But first let us define some variables. Let $A$ denote the adjacency matrix of a random graph generated from the distribution described in Section \ref{section:intro}. And let $A_E = \mathbb{E}[A]$ be the expected adjacency matrix, with entries $a/n$ and $b/n$. Then $A_E$ is a rank two matrix with two non-zero eigenvalues $\lambda_1 = a+b$ and $\lambda_2 = a-b$. Then unit eigenvector $\vect{u_1}$ corresponding to the eigenvalue $a+b$ has coordinates:
\begin{equation}
    \vect{u_1}(i) = \frac{1}{\sqrt{2n}} \forall i = 1,\dots, 2n
\end{equation}

while the unit eigenvector $\vect{u_2}$ corresponding to the eigenvalue $a-b$ has coordinates
\begin{equation}
    \vect{u_2}(i) =
    \begin{cases*}
      \frac{1}{\sqrt{2n}} & if $i \in V_1$ \\
      -\frac{1}{\sqrt{2n}}        & if $i \in V_2$
    \end{cases*}
  \end{equation}

\begin{algorithm}
    \begin{algorithmic}[1]
        \State $A \gets \text{adjacency matrix}$
        \State $d \gets a+b$
        \State $(\lambda_1, \vect{w_1}), (\lambda_2, \vect{w_2}) \gets \text{top two eigenvalues and eigenvectors of $A$} $
        \State $W \gets \text{the subspace spanned by $(\vect{w_1}, \vect{w_2})$} $
        \State $\vect{v_1} \gets \text{projection of $(1,1,\dots,1)$ on to $W$}$
        \State $\vect{v_2} \gets \text{the unit vector in $W$ perpendicular to $\vect{v_1}$}$
        \State $V_1' \subset V \gets \text{highest $n$ vertices in $\vect{v_2}$}$
        \State $V_2' \subset V \gets \text{lowest $n$ vertices in $\vect{v_2}$}$
        \Return $(V'_1,V'_2)$
    \end{algorithmic}
    \caption{Spectral Partition}
    \label{fig:spectral1}
        
\end{algorithm}

Note that $\vect{u_2}$, the second eigenvector of $A_E$, identifies the correct partition. We would like to use the second eigenvector of $A$ to approximate $\vect{u_2}$. The Spectral Algorithm as given in Figure \ref{fig:spectral1} is based on \cite{dharmawan1}'s simplification of the original Spectral Algorithm. The intuition is that $\vect{v_1}$ should be close to $\vect{u_1}$ and $\vect{v_2}$ should be close to $\vect{u_2}$, but it is hard to recover them directly. But we also know that the subspace $W$ is going to be close to the subspace spanned by $\vect{u_1}, \vect{u_2}$, and we know that $\vect{u_1} = (1,1,\dots,1)$, and that $\vect{u_1} \perp \vect{u_2}$. The proof of Theorem  \ref{theorem:chin-theorem} follows that intuition, and amounts to quantifying that notion of "closeness" by bounding the sin of the angle between the subspaces and vectors.  The rest of this section summarizes the proof as shown in \cite{chin2015stochasticblockmodelcommunity} and \cite{dharmawan1}.

\subsection{Bounding the difference matrix}

Let $M = A - A_E$ be the difference between the expected and realized adjacency matrix. Then \cite{dharmawan1} asserts the following:

\begin{theorem}
\label{theorem:M_bound_Chin}
    There exist constants $C_1, C_2$ such that if $a > b > C_1$, and matrix $M$ is obtained as described above, then we have 
    \begin{equation}
        ||M|| \leq C_2 \sqrt{a+b}
    \end{equation}
    with probability $1-o(1)$.
\end{theorem}

The matrix norm $||M||$ in Theorem \ref{theorem:M_bound_Chin} refers to the spectral norm, defined as $\sup\{||Mx||_2 : ||x||_2 \leq 1\}$. For the rest of this paper, all matrix norms also refer to the same definition. 

\subsection{Bounding the angle between subspaces and vectors}

Let $W$ be the two dimensional eigenspace corresponding to the top two eigenvalues of $A$, and let $W_E$ be the corresponding space of $A_E$. Then the following theorem claims that the angle $\angle(W,W_E)$ is sufficiently small. Here we use the usual convention $\sin \angle(W_1, W_2) := ||P_{W_1} - P_{W_2}||$ where $P_W$ is the orthogonal projection onto $W$. By using Davis-Kahan theorem as shown in \cite{DAVIS1963159}, we have Theorem \ref{theorem:W_angle_bound_Chin}

\begin{theorem}
\label{theorem:W_angle_bound_Chin}
    There exist constants $C_1, C_2$ such that if $a > b > C_1$, and subspaces $W,W_E$ are as described above, then we have 
    \begin{equation}
        \sin \angle(W,W_E) \leq C_2 \frac{\sqrt{a+b}}{a-b}
    \end{equation}
    with probability $1-o(1)$.
\end{theorem}

Given the upperbound of the angle between $W$ and $W_E$, the next step is to bound the angle between $\vect{u_2}$ (the second eigenvector of $A_E$) and $\vect{v_2}$ (the vector obtained in step 6 of \textbf{Spectral Partition}). The intuition is that, if these two vectors are close to each other, then $\textbf{Spectral Partition}$ will output a mostly correct partition. It turns out that the lemma in \cite{chin2015stochasticblockmodelcommunity} proved that $\sin \angle(\vect{u_2},\vect{v_2})$ is proportional to the $\sqrt{\sin \angle(W,W_E)}$, from which we have the following result:

\begin{theorem}
\label{theorem:vector_angle_bound_Chin}
    There exist constants $C_1, C_2$ such that if $a > b > C_1$, and vectors $\vect{u_2},\vect{v_2}$ are as described above, then we have 
    \begin{equation}
        \sin \angle(\vect{u_2},\vect{v_2}) \leq C_2 \sqrt{\frac{\sqrt{a+b}}{a-b}}
    \end{equation}
    with probability $1-o(1)$.
\end{theorem}

\subsection{Bounding error from \textbf{Spectral Partition}}
Finally, \cite{dharmawan1} shows that 
\begin{theorem}
\label{theorem:dharmawan_log}
    There exists a constant $C$ such that
    \begin{equation}
        \sin \angle(\vect{u_2},\vect{v_2}) \geq \frac{C}{\sqrt[4]{\log (2/ \gamma)}}
    \end{equation}
    with probability $1-o(1)$.

\end{theorem}
which, combined with Theorem \ref{theorem:vector_angle_bound_Chin}, gives us the following result:

\begin{theorem}
\label{theorem:gamma_bound_quad_Chin}
    There exist constants $C_1, C_2$ such that if $a > b > C_1$, then we have 
    \begin{equation}
        \frac{1}{\sqrt{ \log( 2/ \gamma ) }} \leq C_2 \frac{\sqrt{a+b}}{a-b}
    \end{equation}
    with probability $1-o(1)$.
\end{theorem}

which is equivalent to Theorem \ref{theorem:chin-theorem}.

In Section \ref{section:exp_1} we show that all these bounds are asymptotically sharp for randomly-generated graphs.

\section{Spectral Algorithm With Multiple Samples}
\label{section:spectral_us}

We use the same Spectral Algorithm as in Figure \ref{fig:spectral1}. However, suppose $A_1, \dots, A_m$ are the adjacency matrices realized from $m$ independent samples, we simply run the Spectral Algorithm on $A = (A_1 + \dots + A_m) / m$. 

Now we shall proceed by replicating each step of the proof in Section \ref{section:spectral_Chin} while introducing the effect of $m \geq 1$ samples on the various bounds.

\subsection{Bounding the difference matrix}

Let $M = A-A_E$ be the difference between the expected and realized adjacency matrix, except in this case the realized matrix $A = (A_1 + \dots + A_m)/m$ is the average of $m$ independent samples.

\begin{theorem}
\label{theorem:M_bound_us}
    There exist constants $C_1, C_2$ such that if $a > b > C_1$, and matrix $M$ is obtained as described above, then we have 
    \begin{equation}
        ||M|| \leq C_2 \frac{\sqrt{a+b}}{\sqrt{m}}
    \end{equation}
    with probability $1-o(1)$.
\end{theorem}

\begin{proof}
Compare Theorem \ref{theorem:M_bound_us} to Theorem \ref{theorem:M_bound_Chin}. Each matrix $A_k, 1 \leq k \leq m$ has entries $A_{ijk}$ that are sampled from a Bernoulli distribution with success probability $p_{ij}$ where $p_{ij} = a/n$ if $i,j$ belong to the same community, and $p_{ij} = b/n$ otherwise. Therefore, the entries of matrix $M$ have mean zero and variance $\sigma_{ij}^2 = \frac{p_{ij}(1-p_{ij})}{m} \leq \sigma^2 $ where $\sigma^2$ is the maximum variance of a single element.

Because $b < a < n/2$ we have:
\begin{equation}
    \sigma_{ij}^2 \leq \sigma^2 = \max \left( \frac{\frac{a}{n}(1-\frac{a}{n})}{m}, \frac{\frac{b}{n}(1-\frac{b}{n})}{m} \right) = \frac{a}{mn} \left( 1 - \frac{a}{n}\right) \leq \frac{a+b}{mn}
\end{equation}

Let $\lambda_1(M)$ be the largest eigenvalue of $M$. Because $M$ is real-valued and symmetric, $\lambda_1(M) = ||M||$. Now we use the result from \cite{Fredi1981TheEO} to determine $\mathbb{E}[\lambda_1(M)]$. Since all entries have mean zero and variance at most $\sigma^2$, we have:

\begin{equation}
    \label{eq:furedi}
    \mathbb{E}[\lambda_1(M)] = 2 \sigma \sqrt{n} + O(n^{1/3} \log n)
\end{equation}

For large enough $n$, the first term dominates. So $\mathbb{E}[\lambda_1(M)] = O(\sigma \sqrt{n})$. Note: \cite{Fredi1981TheEO} uses the premise that all entries have mean zero and common variance, but \cite{krivelevich2000concentrationeigenvaluesrandomsymmetric} showed that the assumption of common variance can be relaxed to $Var[M_{ij}] \leq \sigma^2$.

Next, also according to \cite{krivelevich2000concentrationeigenvaluesrandomsymmetric}, there are positive constants $c$ and $K$ such that for any $t > K$,

\begin{equation}
    \label{eq:krivelevich}
    P \left[ | \lambda_1(M) - \mathbb{E}[ \lambda_1(M) ] | \geq t \right] \leq e^{-ct^2}
\end{equation}

Combining equations \ref{eq:furedi} and \ref{eq:krivelevich}, we see that the probability of $||M||$ deviating from the mean $2\sigma \sqrt{n}$ is vanishingly small the further it is from the mean. So if we take any constant $C_2 > 2$, with probability $1-o(1)$, $||M||$ will not exceed $C_2 \sigma \sqrt{n}$. In other words, for large enough $b$ (and consequently $a,n$), we have with probability $1-o(1)$:

\begin{equation}
    ||M|| \leq C_2 \sigma \sqrt{n} \leq C_2 \frac{\sqrt{a+b}}{\sqrt{mn}} \sqrt{n}
\end{equation}
which completes the proof for Theorem \ref{theorem:M_bound_us}.
\end{proof}

\subsection{Bounding the angle between subspaces and vectors}
Our version of Theorems \ref{theorem:W_angle_bound_Chin} and \ref{theorem:vector_angle_bound_Chin} are as follows:

\begin{theorem}
\label{theorem:W_angle_bound_us}
    There exist constants $C_1, C_2$ such that if $a > b > C_1$, then we have 
    \begin{equation}
        \sin \angle(W,W_E) \leq C_2 \frac{\sqrt{a+b}}{\sqrt{m}(a-b)}
    \end{equation}
    with probability $1-o(1)$.
\end{theorem}

\begin{theorem}
\label{theorem:vector_angle_bound_us}
    There exist constants $C_1, C_2$ such that if $a > b > C_1$, then we have 
    \begin{equation}
        \sin \angle(\vect{u_2},\vect{v_2}) \leq C_2 \sqrt{\frac{\sqrt{a+b}}{\sqrt{m}(a-b)}}
    \end{equation}
    with probability $1-o(1)$.
\end{theorem}

These two theorems follow immediately because $\sin \angle(W,W_E)$ is linearly proportional to $||M||$ (Davis-Kahan theorem \cite{DAVIS1963159} ) and that $\sin \angle(\vect{u_2},\vect{v_2})$ is proportional to $\sqrt{\sin \angle(W,W_E)}$ as shown in \cite{chin2015stochasticblockmodelcommunity}.

\subsection{Bounding error from \textbf{Spectral Partition}}

When we combine Theorem \ref{theorem:dharmawan_log} and Theorem \ref{theorem:vector_angle_bound_us}, we get the following:
$$\frac{C_1}{\sqrt[4]{\log (2/ \gamma)}} \leq \sin \angle(\vect{u_2},\vect{v_2}) \leq C_2 \sqrt{\frac{\sqrt{a+b}}{\sqrt{m}(a-b)}}$$
for some constants $C_1,C_2$. This simplifies to:
$$\frac{(a-b)^2}{a+b} \leq \frac{C_3}{m} \log \frac{2}{\gamma}$$
for some constant $C_3$. This form is precisely the claim of Theorem \ref{theorem:main-theorem}. The inequality sign is reversed because in Theorem \ref{theorem:main-theorem}, $\gamma$ denotes the desired upper bound of the error-rate that the algorithm can guarantee, whereas in the inequality above, $\gamma$ denotes the \textbf{realized} error rate after the algorithm is completed.

\begin{figure}[h]
\centering
\begin{minipage}{0.45\textwidth}
    \centering
    \includegraphics[width=\textwidth]{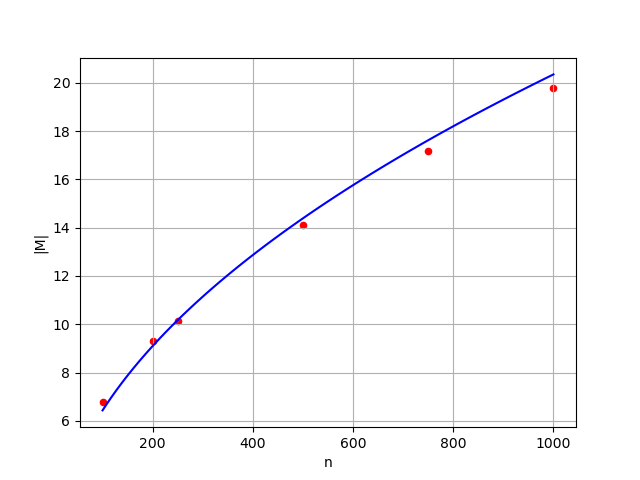}
    \caption{$||M||$ (red dots) versus $C\sqrt{a+b}$ (blue curve)}
    \label{fig:M_bound_1}
\end{minipage}
\hfill
\begin{minipage}{0.45\textwidth}
    \centering
    \includegraphics[width=\textwidth]{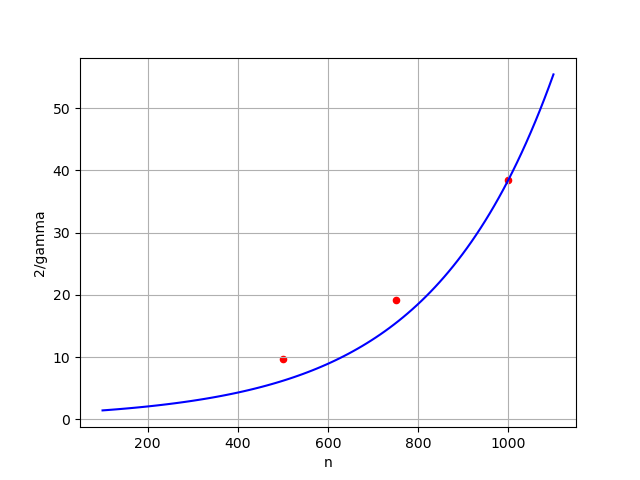}
    \caption{$2/\gamma$  and $e^{C(a-b)^2/(a+b)}$  for various $n$}
    \label{fig:gamma_bound_1}
\end{minipage}
\end{figure}

\begin{figure}[h]
    \begin{minipage}{0.3\textwidth}
    \includegraphics[width=\linewidth]{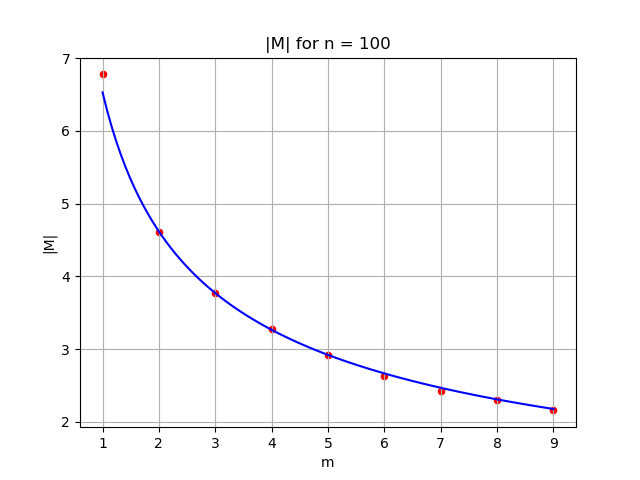}
    \end{minipage}
    \hspace{\fill} 
    \begin{minipage}{0.3\textwidth}
    \includegraphics[width=\linewidth]{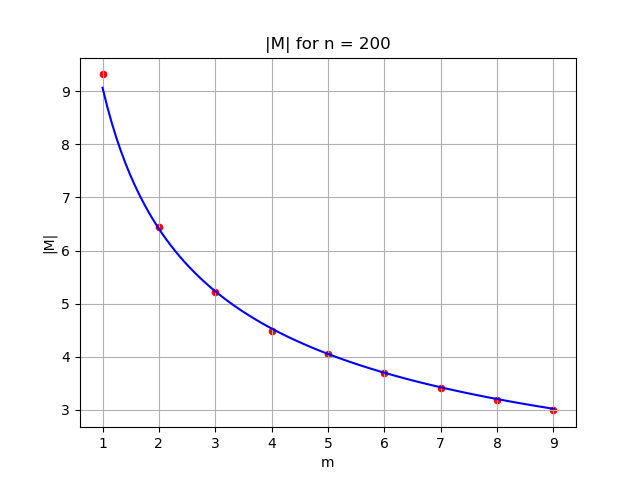}
    \end{minipage}
    \hspace{\fill} 
    \begin{minipage}{0.3\textwidth}
    \includegraphics[width=\linewidth]{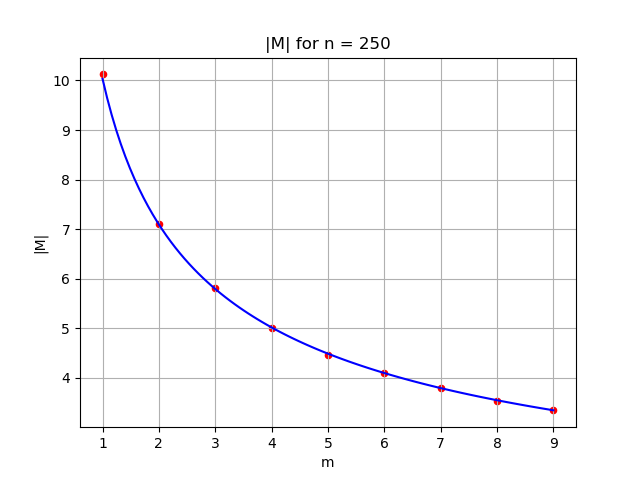}
    \end{minipage}
    
    \begin{minipage}{0.3\textwidth}
    \includegraphics[width=\linewidth]{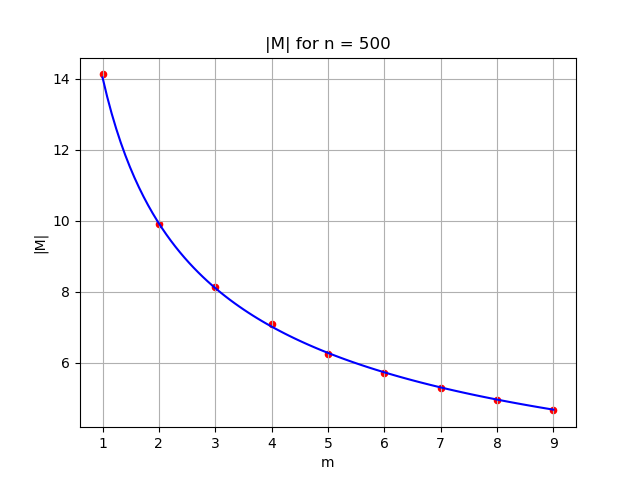}
    \end{minipage}
    \hspace{\fill} 
    \begin{minipage}{0.3\textwidth}
    \includegraphics[width=\linewidth]{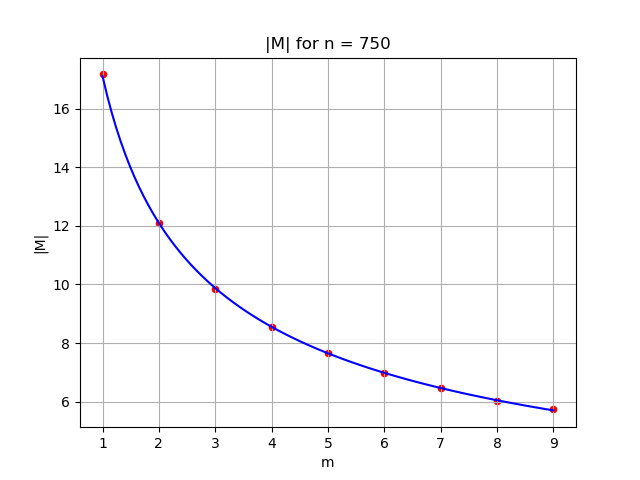}
    \end{minipage}
    \hspace{\fill} 
    \begin{minipage}{0.3\textwidth}
    \includegraphics[width=\linewidth]{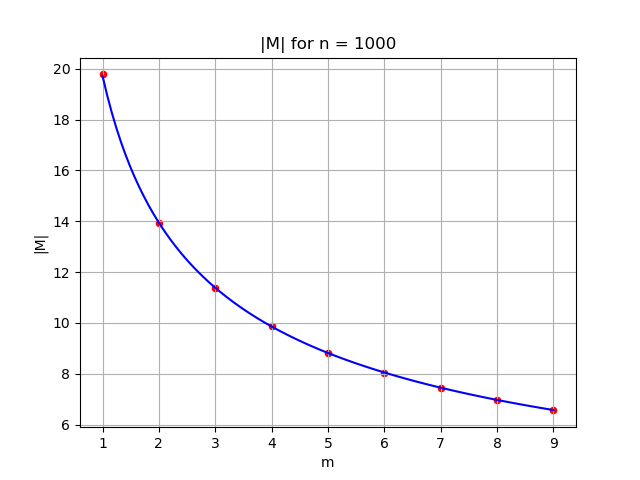}
    \end{minipage}

\caption{$||M||$ (red dots) versus $C / \sqrt{m}$ (blue curves) for various values of $n$} \label{fig:M_bound_m}
\end{figure}

In Section \ref{section:exp_m} we show that the bounds in Theorems \ref{theorem:M_bound_us}, \ref{theorem:W_angle_bound_us}, and \ref{theorem:main-theorem} are asymptotically sharp.

\section{Experimental Results for 1 Sample}
\label{section:exp_1}

We conducted the experiments for 1 sample, with $a=0.06n,b=0.04n$ for various values of $n$. For each value of $n$, we conducted 500 experiments, each with a different random seed. The plots are available in the appendix. The raw CSV file containing the data, as well as the code to reproduce the data and plot, is included in the ZIP file as supplementary material and is also available upon request.

First, we show that the bound in Theorem \ref{theorem:M_bound_Chin} is sharp, even after we skip the deletion step. Indeed, Figure \ref{fig:M_bound_1} shows the maximum $||M||$ compared against the curve $0.644\sqrt{n}$. Note that when $a/n$ and $b/n$ are constant, the RHS of Theorem \ref{theorem:M_bound_Chin} scales with $\sqrt{n}$. 

Now, we show that the bound in Theorem \ref{theorem:chin-theorem} is sharp. For each value of $n$, we take the worst performance (largest $\gamma$) out of the 500 experiments. For lower values of $n$, the maximum $\gamma$ was 0.5, so we did not include those values of $n$, because for 2 communities, $0.5$ is the worst possible performance. We are left with only 3 values of $n$, namely $n=500,750,1000$. Figure \ref{fig:gamma_bound_1} shows $2/\gamma_{max}$ compared against the curve $\exp \left(C \frac{(a-b)^2}{a+b} \right)$. Later in Section \ref{section:exp_m} we will see that for some higher values of $m$, when we keep $m$ constant, we see the relationship more clearly between $n$ and $\gamma$, since at that point more values of $n$ produce meaningfully small $\gamma$.

\section{Experimental Results for Multiple Samples}
\label{section:exp_m}

We repeated the experiments in the previous section, keeping the same values of $n,a,b$, but for $m=1,2,\dots, 9$ samples, each time running 500 experiments. 

In Figure \ref{fig:M_bound_m} we see that for a fixed $n,a,b$, the resulting $||M||$ decreases with $\sqrt{m}$, showing that the bound in Theorem \ref{theorem:M_bound_us} is sharp. 

\begin{figure}[h]
    \begin{minipage}{0.22\textwidth}
    \includegraphics[width=\linewidth]{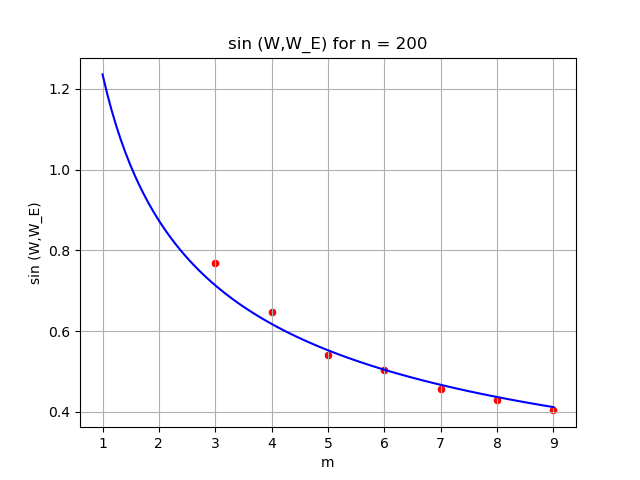}
    \end{minipage}
    \hspace{\fill} 
    \begin{minipage}{0.22\textwidth}
    \includegraphics[width=\linewidth]{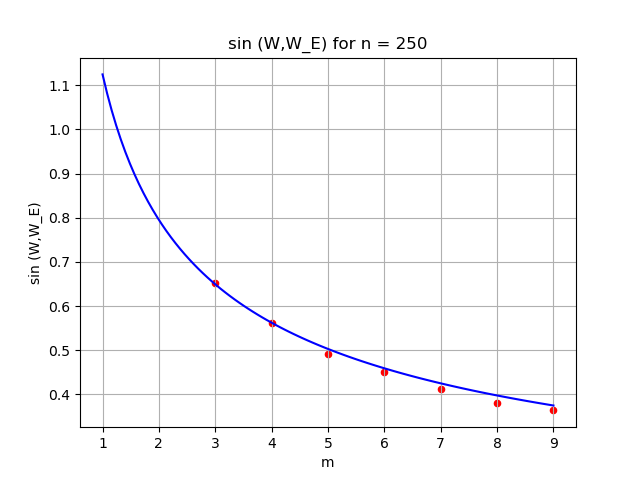}
    \end{minipage}
    \hspace{\fill} 
    \begin{minipage}{0.22\textwidth}
    \includegraphics[width=\linewidth]{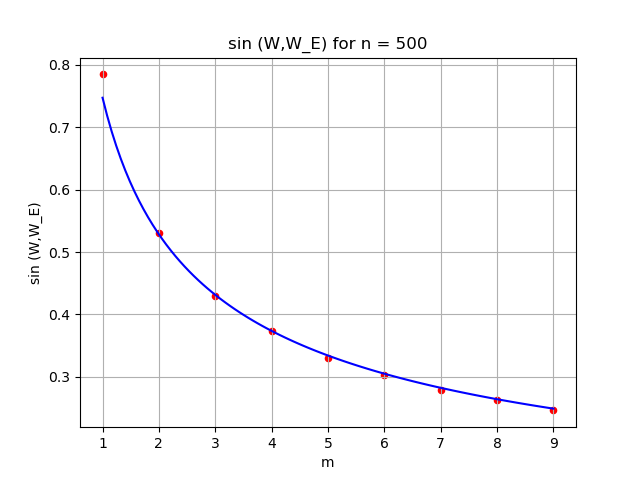}
    \end{minipage}
    \hspace{\fill} 
    \begin{minipage}{0.22\textwidth}
    \includegraphics[width=\linewidth]{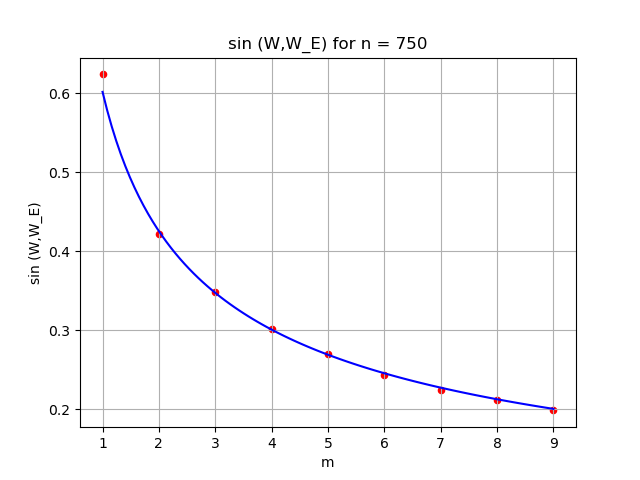}
    \end{minipage}
\caption{$\sin \angle(W,W_E)$ (red dots) versus $C / \sqrt{m}$ (blue curves) for various $n$} \label{fig:sin_W_bound_m}
\end{figure}

\begin{figure}[h]
    \begin{minipage}{0.3\textwidth}
    \includegraphics[width=\linewidth]{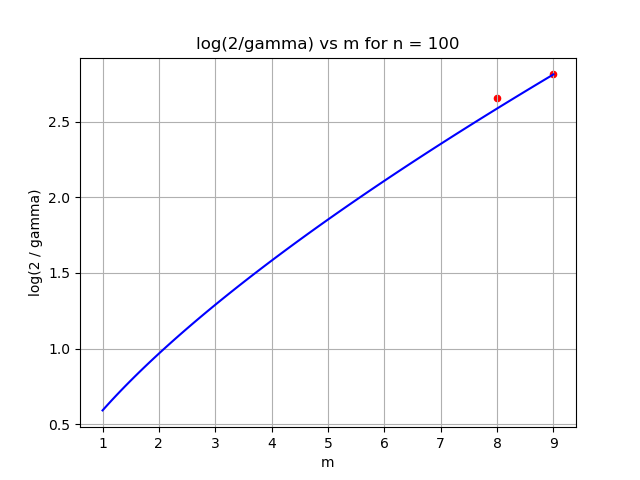}
    \end{minipage}
    \hspace{\fill} 
    \begin{minipage}{0.3\textwidth}
    \includegraphics[width=\linewidth]{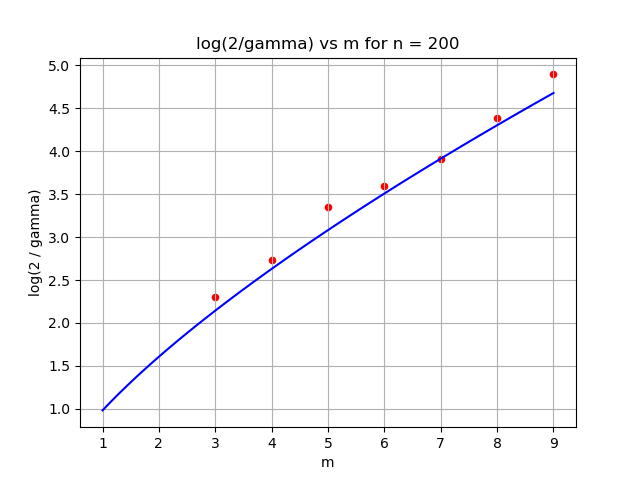}
    \end{minipage}
    \hspace{\fill} 
    \begin{minipage}{0.3\textwidth}
    \includegraphics[width=\linewidth]{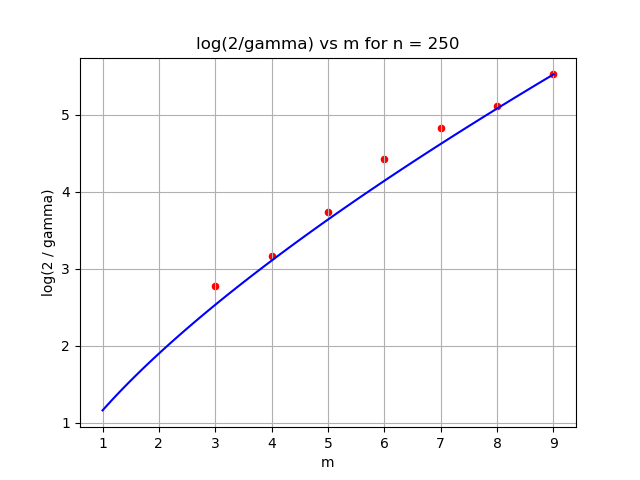}
    \end{minipage}

    \begin{minipage}{0.3\textwidth}
    \includegraphics[width=\linewidth]{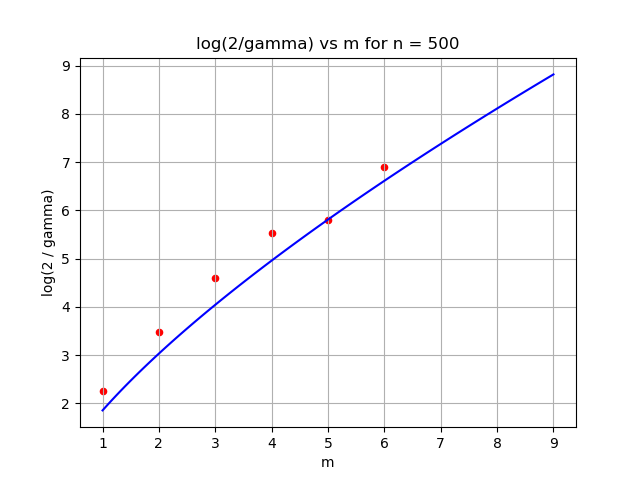}
    \end{minipage}
    \hspace{\fill} 
    \begin{minipage}{0.3\textwidth}
    \includegraphics[width=\linewidth]{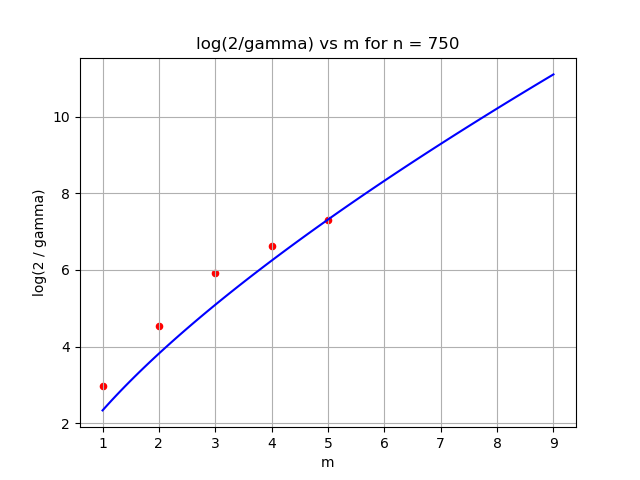}
    \end{minipage}
    \hspace{\fill} 
    \begin{minipage}{0.3\textwidth}
    \includegraphics[width=\linewidth]{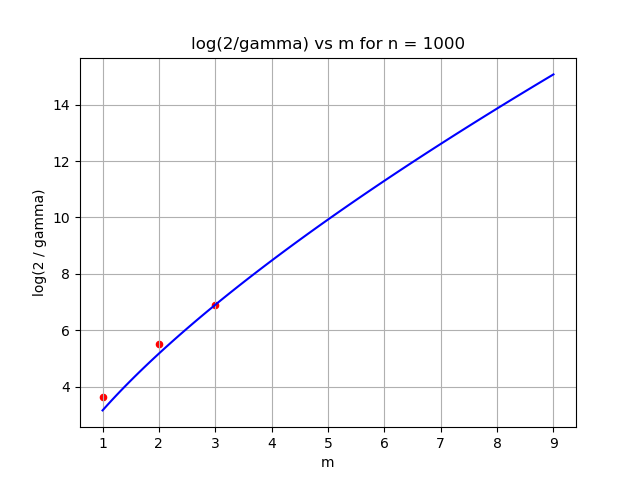}
    \end{minipage}

\caption{$\log(2 / \gamma)$ (red dots) versus $C / m$ (blue curves) for various $m,n$} \label{fig:gamma_bound_m}
\end{figure}

\begin{figure}[h]
    \begin{minipage}{0.3\textwidth}
    \includegraphics[width=\linewidth]{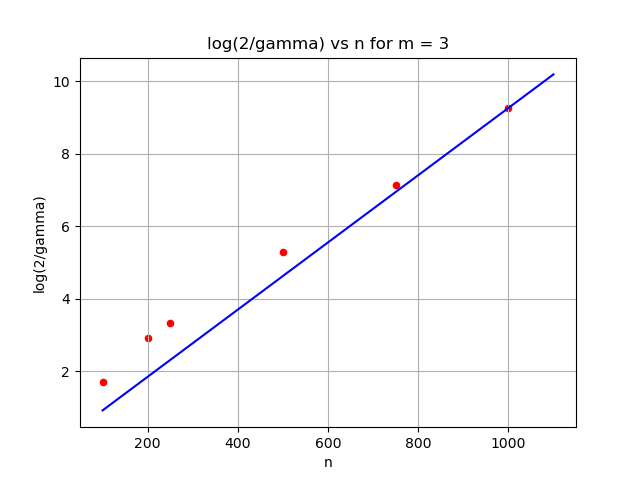}
    \end{minipage}
    \hspace{\fill} 
    \begin{minipage}{0.3\textwidth}
    \includegraphics[width=\linewidth]{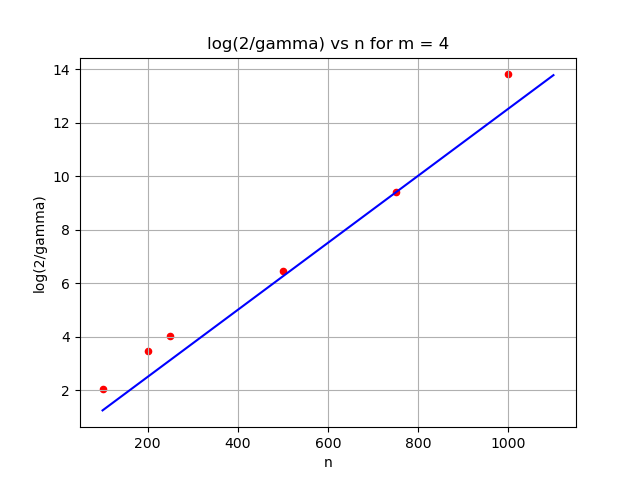}
    \end{minipage}
    \hspace{\fill} 
    \begin{minipage}{0.3\textwidth}
    \includegraphics[width=\linewidth]{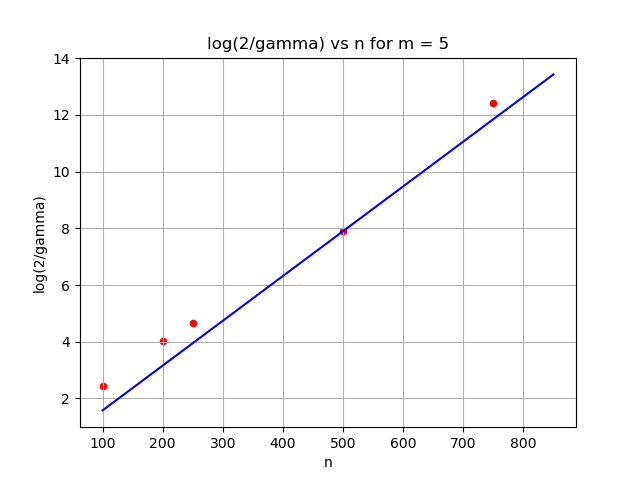}
    \end{minipage}

    \begin{minipage}{0.3\textwidth}
    \includegraphics[width=\linewidth]{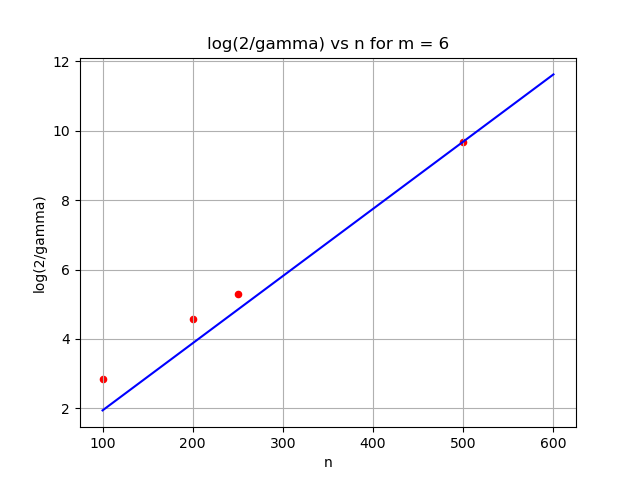}
    \end{minipage}
    \hspace{\fill} 
    \begin{minipage}{0.3\textwidth}
    \includegraphics[width=\linewidth]{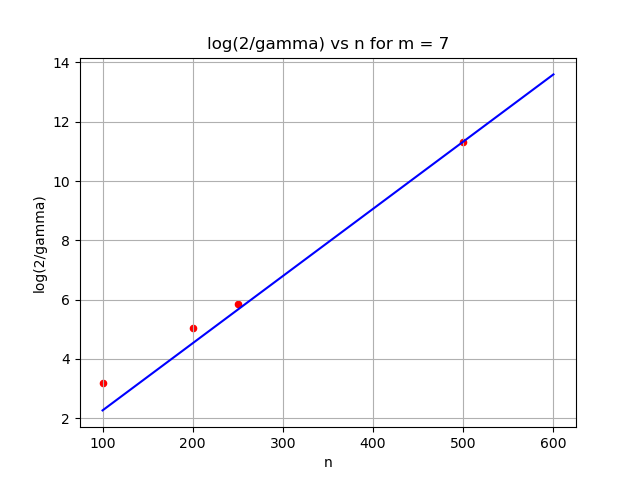}
    \end{minipage}
    \hspace{\fill} 
    \begin{minipage}{0.3\textwidth}
    \includegraphics[width=\linewidth]{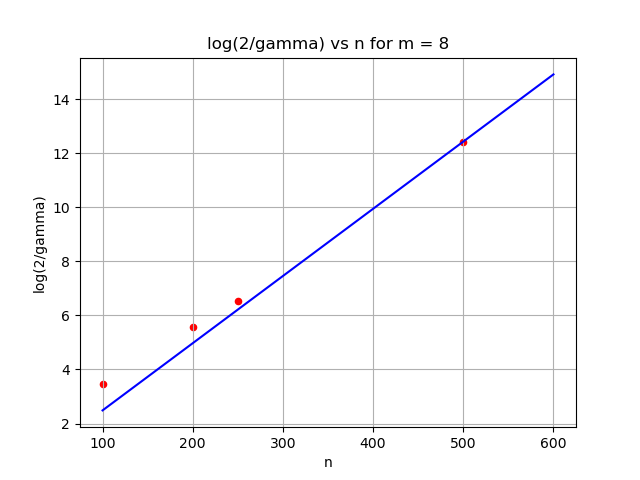}
    \end{minipage}

\caption{$\log(2/\gamma)$ versus $C(a-b)^2 / (a+b)$ for fixed values of $m \in \{3,4,5,6,7,8 \}$} \label{fig:gamma_bound_for_each_m}. The case $m=9$ gave $\gamma = 0$ so $\log(2 / \gamma$ is not plottable.
\end{figure}

In Figure \ref{fig:sin_W_bound_m} we also see that the resulting $\sin \angle(W,W_E)$ decreases with $\sqrt{m}$, showing that the bound in Theorem \ref{theorem:W_angle_bound_us} is sharp. In this plot, we omitted some of the lower values of $n$ and $m$ because $\sin \angle(W,W_E)$ was close to 1 for those combinations. Later we will also see that these combinations correspond to large $\gamma$. But even when, for a given $n$, the lower values of $m$ are discarded, the rest of the data still follow an inverse-square-root relationship, proving the sharpness of the bound.

Now we show the sharpness of bound in Theorem \ref{theorem:main-theorem}. Figure \ref{fig:gamma_bound_m} shows the plot of $\log(2 / \gamma)$ against $C/m$. As before, we discard the experiments where $\gamma$ is close to 0.5, primarily in lower values of $n$ and $m$ (the 
"harder" problems). But this time, we also discard $\gamma=0$ because for higher values of $n,m$, in all 500 experiments for that setting, we were able to achieve perfect recovery of the partition. Note that in general, the red dots match the blue curves reasonably well. However, sometimes the red dots fall above the blue lines. Since the theoretical bound is an upper bound on $\gamma$, it is expected that the observed worst-case $\gamma$ will often fall below the predicted curve, particularly for larger $n$ and $m$.

This observation reinforces the result in Theorem \ref{theorem:main-theorem} that as $m,n$ increase, and we hold $a/n, b/n$ constant, the problem becomes "easier" and the upper bound of $\gamma$ dramatically decreases. 

In particular, for $n=100$, then $b=4$ and we see that $\gamma=0.5$, practically no better than random guessing. However, with as low as $m=3$ samples, we start seeing that $\gamma \approx 0.36$, and it keeps improving until we get $\gamma = 0.05$ with $m=9$ samples. 

Finally, we hold $m$ constant and show the relationship between $2 / \gamma$ and $\exp \left(C \frac{(a-b)^2}{a+b} \right)$ for various values of $n$ (i.e. replicating Figure \ref{fig:gamma_bound_1} but this time we have more values of $n$ that are not discarded). We see that, for a fixed $m$, Theorem \ref{theorem:main-theorem} replicates Theorem \ref{theorem:chin-theorem} and demonstrates that the inverse-log relationship is indeed sharp. 

\section{Conclusion and Future Work}
In this work, we have developed an improved spectral algorithm for community detection in the two-block stochastic block model when multiple samples from the same underlying random graph distribution are available. Our theoretical analysis establishes a relationship between the number of samples and the recovery rate, demonstrating that additional samples significantly improve detection performance.

The sharp bounds that we have derived characterize the fundamental limits of our spectral approach, providing practitioners with clear guidelines on the sample complexity required to achieve the desired recovery accuracy. Our experimental results across various parameter regimes validate these theoretical guarantees, confirming that the performance improvements scale as predicted by our analysis.

This work opens several promising directions for future research. Extensions to models with multiple blocks, heterogeneous edge densities, or partially overlapping community structures would broaden the applicability of our approach. Additionally, investigating whether alternative algorithmic paradigms can overcome the limitations of spectral methods in the multi-sample setting remains an interesting open question. 

All code necessary to reproduce the results will be published to the public repository: \url{https://github.com/hendrata-th/sbm-multiple-samples}

%
%
%
 \bibliographystyle{splncs04}
 \bibliography{mybibfile}

\end{document}